\documentclass{article}
\usepackage[utf8]{inputenc}
\usepackage[english]{babel}
\usepackage[T1]{fontenc}
\usepackage{amsmath}
\usepackage{amssymb}
\usepackage[standard]{ntheorem}
\usepackage{stmaryrd}
\usepackage{xcolor}
\usepackage{enumitem}
\usepackage{graphicx}
\usepackage{geometry}
\usepackage{csquotes}
\usepackage{biblatex}
\usepackage{enumitem}
\usepackage{bm}
\newtheorem*{thm}{Theorem}
\newtheorem*{propo}{Proposition}
\newtheorem*{assumption}{Assumption}

\newcommand{\R}{\mathbb{R}}
\newcommand{\Z}{\mathbb{Z}}

\title{Optimal Subgraph on Disturbed Network}
\author{Matthieu Guillot, El-Houssaine Aghezzaf, Nour-Eddin El Faouzi, Angelo Furno}
\date{\today}

\begin{document}

\maketitle

\section{Introduction}
\label{intro}

Public transportation networks in large urban areas are supporting a very large number of passengers everyday. Problems arise on a daily basis, causing more or less disruptions depending on the seriousness of the events. In most cases, the consequences are local (delays of public transportation passages, local traffic jams, etc...). Thus, even if the locality of the disturbance can be relatively large and cause significant delays, the control techniques make possible to restore a fluid state of the network at the end of the day. However, the current global pandemic linked to COVID-19 reminds us that in extreme cases, everyone's habits can be changed drastically \cite{Tirachini2020}. In particular, the parameters defining urban transportation networks have been completely turned upside down \cite{liu2020}. On the one hand, during lockdown, the users of the network adopt completely different habits: some no longer use it (total or partial telework) and some would rather use individual types of transportation more that than collective ones. Thus, the demand linked to the transportation network has drastically changed both qualitatively and quantitatively \cite{di2020}. On the other hand, the offer has also been modified. The demand's change and the right of withdrawal of the agents are part of the offer's change \cite{wilbur2020}. All the pre-existing forecasts (transit times and frequencies, start and end times of service, connections, etc.) have become obsolete. Obviously, all these modifications were experienced in practice during the lockdown that occurred following the progression of the pandemic COVID-19, but one can easily imagine that in the next decades or even the next years, other events will tend to disrupt if not all urban networks but at least part of them.

In this article, we consider an urban zone and two kinds of networks on it: a {\em urban network (UN)} which represent a grid of the urban areas and a {\em public transportation network (PTN)} which represent the current existing buses network of the corresponding urban areas. We assume that we know the traveling time on the transportation network: the minimal time between any pair of nodes of PTN, and all the access times between any node of the urban network and any node of the transportation network. We also assume that we know the modified demand, that we assume to be lower than usually (which is the case during the current pandemic). A solution to our problem is a subnetwork of the transportation network that guarantee that: (i) the minimal access time from any node of the urban network to the new network is not {\em too large} compared to the original transportation network; (ii) for any itinerary, the delay caused by the deletion of nodes of the transportation network is not {\em too big}; and (iii) the number of nodes of the transportation network has been reduced at least by a known factor. A solution is optimal if it induces a minimal global delay.

\section{Related Work}
\label{relatedWork}

Several studies of the impact of COVID-19 on public transportation systems have begun to emerge. For the impact on the transit frequency, Gkiotsalitis et al. give a model that provides optimal vehicle redistribution accross metro lines of Washington DC for different scenarios based on different social distances rules \cite{gkiotsalitis2021}. Dakic et al. model and develop an optimization tool to determine the optimal bus frequencies and vehicle allocation to reduce the operating cost of the network. Regarding the sanitary conditions and the contamination exposures, Jia et al. identify the `key stations' of the railway network of Beijin to avoid in order to limit the risks of contamination for the passengers \cite{jia2021}. With a strategic point of view, Wang et al. study the effective policies for reopening phase. These policies are biased on the work-for-home, the traffic and sanitary conditions, but also on the transit capacity and demand \cite{wang2021}.

Network design is a big issue too, especially in big urban areas. LeBlanc defines the network design to find the optimal frequencies of transit lines \cite{leblanc1988}. Lee et al. do the same kind of work, but with variable demand \cite{lee2005}. More recently Cipriani, Gori and Petrelli give case study of a resolution of network design problem in the city of Rome, with multimodal properties and complex road network topology \cite{cipriani2012}. 
The highlighting of an efficient subgraph has also been studied in transportation applications. Arbex and da Cunha are interested in finding an efficient subgraph and the corresponding frequency using genetic algorithms. The objective in this article is to optimize both passenger's and operator's costs \cite{arbex2015}.

\section{Definition of the Problem}
\label{sec:definitionProblem}
Let $PTN = (V,E_t)$ be an undirected graph representing the {\em Public transportation Network}, with $|V| = n_t$ nodes and $|E_t| = m_t$ edges. The nodes $V$ represent the bus stops and the edges $E_t$ the possible links between the corresponding bus stops. Let $c : E_t \mapsto \R^+$ a cost function over the edges of $PTN$, which represent the {\em traveling time}. We assume that we know the matrix $PC\in {\R^+}^{n_t \times n_t}$ of the shortest paths in $PTN$, which means that $PC[v_1][v_2]$ is the traveling time of the shortest path between $v_1$ and $v_2$ for each $v_1, v_2 \in V$. Figure \ref{fig:PTN} is an example for six bus stops.

\begin{figure}[ht]
    \centering
    \includegraphics[width=10cm]{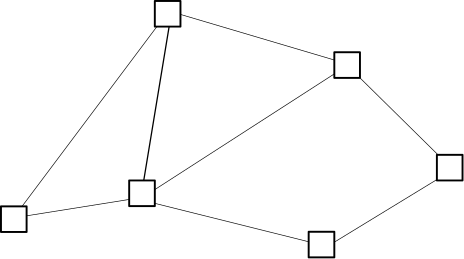}
    \caption{PTN for six bus stops}
    \label{fig:PTN}
\end{figure}

Let $UN = (U \cup V, E_u)$ be a complete bipartite graph representing the {\em Urban Network}. The set of nodes $U$ are the $n_u = |U|$ centroids of urban areas, which represent the possible origins and destinations of the demand. $E_u$ represent the possible links between the urban areas and the bus stops. As $UN$ is complete, the users are theoretically able to walk to any bus stop. Let $d : U \times V \mapsto \R^+$ be the {\em access time}: $d(u,v)$ represent the average time that the users have to walk to go from the urban area $u \in U$ to the bus stop $v \in V$. For any $u\in U$, we denote by $d_{acc}(u) = \min_{v\in V}{d(u, v)}$, and by $acc(u) = argmin_{v \in V}{d(u,v)}$. Let $OD \in {\Z^+}^{n_u \times n_u}$ the origin/destination matrix: $OD[u_1][u_2]$ is the number of users who wish to go from urban area $u_1 \in U$ to $u_2 \in U$. Figures \ref{fig:PTNetUN} and \ref{fig:PTNetUNetAccess} are the graphical representation of PTN, UN and the corresponding $d_{acc}$ and $acc$.

\begin{figure}[ht]
\begin{minipage}{0.50\linewidth}
    \includegraphics[width=8cm]{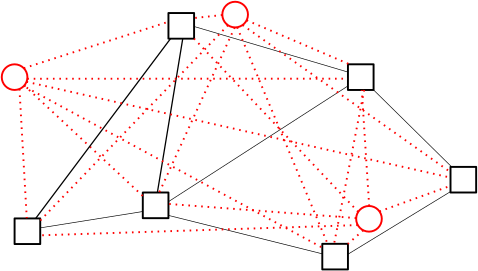}
    \caption{PTN (in black) and UN (in red)}
    \label{fig:PTNetUN}
\end{minipage}
\begin{minipage}{0.35\linewidth}
    \includegraphics[width=8cm]{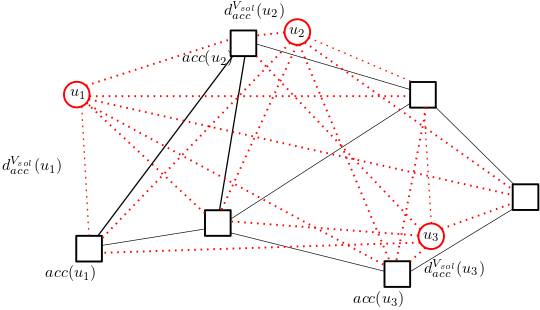}
    \caption{PTN, UN and the corresponding $d_{acc}$ and $acc$}
    \label{fig:PTNetUNetAccess}
\end{minipage}
\end{figure}

An instance of our problem is a tuple $I = (PTN, c, PC, UN, d, OD, p_{elim}, \alpha, k)$, where $PTN, c, PC, UN, d$ and $OD$ are defined as above. Moreover:

\begin{itemize}
    \item $p_{elim} \in [0,1]$ is the minimum percentage of the bus stop that has to be deleted
    \item $1 < \alpha \in \R^+$  is the admissible increase factor of the delay in the new network we want to design for each pair origin/destination
    \item $k\in {\R^+}^{n_u}$, where $k[u]$ is the admissible increase factor for the access time of $u\in U$
\end{itemize}

A solution of our problem is a subset $V_{sol}$ of $V$. For a solution $V_{sol} \subseteq V$ and an origin/destination $(u_1, u_2) \in U$, the optimal travel time from $u_1$ to $u_2$ is defined as:
\begin{equation}
Opt_{V_{sol}}(u_1, u_2) = \min_{(v_1, v_2) \in V_{sol}^2}{d(u_1, v_1) + PC[v_1][v_2] + d(u_2, v_2)}
\label{optVSol}
\end{equation}

We also define the total weighted traveling time of a solution $V_{sol}$ as $TWT_{V_{sol}} = \sum\limits_{(u_1, u_2)\in U^2}{OD[u_1][u_2]Opt_{V_{sol}}(u_1, u_2)}$

In particular, in the original network PTN, the optimal traveling time between $u_1$ and $u_2$ is 
\begin{equation}
Opt_{V}(u_1, u_2) = \min_{(v_1, v_2) \in V^2}{d(u_1, v_1) + PC[v_1][v_2] + d(u_2, v_2)}
\label{optV}
\end{equation}

and the total weighted traveling time of PTN is $TWT_{V} = \sum\limits_{(u_1, u_2)\in U^2}{OD[u_1][u_2]Opt_{V_{sol}}(u_1, u_2)}$

Finally, we define, for all $u\in U$, $d_{acc}^{V_{sol}}(u)$ as $d_{acc}^{V_{sol}}(u) = \min\limits_{v\in V_{sol}}{d(u,v)}$. In particular, $d_{acc} = d_{acc}^{V}$.

As $|V| \geq |V_{sol}|$, we have $TWT_{V} \leq TWT_{V_{sol}}$ and, as a consequence, the total weighted delay induced by the choice of solution $V_{sol}$ ($i.e.$ the deletion of all the nodes in $V\setminus V_{sol}$) is $TWT_{V} - TWT_{V_{sol}} \geq 0$. Figure \ref{fig:PTNetVSol.png} represents a solution and the corresponding $d_{acc}^{V_{sol}}$

\begin{figure}[ht]
    \centering
    \includegraphics[width=10cm]{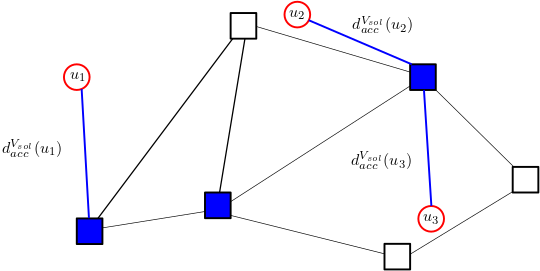}
    \caption{In blue: $V_{sol}$ and the corresponding $d_{acc}^{V_{sol}}$}
    \label{fig:PTNetVSol.png}
\end{figure}

\begin{assumption}
We assume that the delay due to the deletion of bus stops is caused by the difference of access time, and not by the difference of shortest path in the bus network. Thus, we have:

\begin{equation}
Opt_{V_{sol}}(u_1, u_2) = \min_{(v_1, v_2) \in V^2}{d(u_1, v_1) + PC[acc(u_1)][acc(u_2)] + d(u_2, v_2)}
\label{optVSolAss}
\end{equation}

Note that with this assumption, we have 

\begin{equation}
    Opt_{V_{sol}}(u_1, u_2) = d_{acc}^{V_{sol}}(u_1) + d_{acc}^{V_{sol}}(u_2) + PC[acc(u_1)][acc(u_2)]
    \label{optVsolAssumptionSimpl}
\end{equation}

\end{assumption}

This assumption is quite strong, because we omit the difference of travel time between the original network and the choice induced by $V_{sol}$. However, this simplification is realistic as most people would accept an increase of the travel time but not of the access time. Moreover, it will imply a very interesting computational benefit.

We want our solution $V_{sol}$ to have some properties 
 \begin{enumerate}[label=\Alph*)]
    \item the access time increase from $u\in U$ induced by the choice of $V_{sol}$ must not increase more than by a factor $k[u]$
    \item the delay induced by the choice of $V_{sol}$ must not increase more than by a factor $\alpha$

    \item the percentage of deletion of $V_{sol}$ has to be at least $p_{elim}$
 \end{enumerate}

More formally, we want $V_{sol}$ to satisfy:
\begin{enumerate}[label=\Alph*)]
     \item $\forall u \in U$, $d_{acc}^{V_{sol}}(u) \leq k[u]\times d_{acc}(u)$
     \item $\forall (u_1, u_2) \in U^2$, $Opt_{V_{sol}}(u_1, u_2) \leq \alpha \times Opt_{V}(u_1, u_2)$ 
    \item $|V_{sol}| \leq (1-p_{elim})|V|$
 \end{enumerate}

    A solution $V_{sol}$ to our problem is said to be {\em feasible} if it satisfies the constraints above. Let us call $\mathcal{V_F}$ the the set of all feasible solutions. As there is a finite number of bus stops, we know that $|\mathcal{V_F}|$ is finite too. Our goal is to find a solution that minimizes the total weighted traveling time. So $V^*_{sol}$ is said to be {\em optimal} if it verifies $TWT_{V^*_{sol}} = \min\limits_{V \in \mathcal{V_F}}{TWT_V}$. The optimal subgraph on disturbed network problem (OSDNP for short) is the problem which consists of finding such a solution.
    
    In order to find a formulation for our problem, as we want to minimize the total weighted traveling time under certain constraints, it is quite natural to consider in a first place mathematical programming, and especially linear programming.

\section{Mixed-Integer Linear Programming Formulation}
\label{MILP}

Let us define, for all $v\in V$ and all $k \in {\R^+}^{n_u}$, $D_u^k = \{v \in V | d(u, v) \leq k[u]d_{acc}(u)\}$. We also define $M$ as a upper bound on the $d(u,v)$. 

Let us consider the mathematical program $(P)$:

\small{
$$
\begin{array}{llllll}
\label{PL}
    min &  \sum\limits_{(u_1, u_2) \in U}{OD[u_1][u_2](d_{acc}^{x}(u_1) + d_{acc}^{x}(u_2))} && \\
    s.t. & \sum\limits_{v \in D_u^k}{x(v)} & \geq & 1  & \forall u \in U & (1)\\
    & d_{acc}^{x}(u) & = & \min\limits_{v \in D_u^k}{d(u, v) + (1-x(v))M} & \forall u \in U & (2)\\
    & d_{acc}^{x}(u_1) + d_{acc}^{x}(u_2) + PC[acc(u_1)][acc(u_2)] & \leq & \alpha(d_{acc}(u_1) + d_{acc}(u_2) + PC[acc(u_1)][acc(u_2)]) & \forall (u_1, u_2) \in U^2 & (3)\\
    & \sum\limits_{v \in V}{x(v)} & \leq & (1-p_{elim})n_t && (4)\\
    & x  & \in & {\Z^+}^{n_t} &&\\
    & d_{acc}^{x} & \in & {\R^+}^{n_u} && 
    
\end{array}
$$
}

Let us call $\mathcal{P_F}$ the set of feasible solutions of $(P)$.

\begin{propo}
$(P)$ is a mixed integer linear program.
\end{propo}

\begin{proof}
As the objective function, and the constraints $(1)$, $(3)$ and $(4)$ are linear with regards to the decision variables, the only thing to show is the linearity of constraints $(2)$.

Let us take one particular $\hat{u} \in U$. We will prove that the constraint $d_{acc}^{x}(\hat{u})  =  \min\limits_{v \in D_{\hat{u}}^k}{d(\hat{u}, v) + (1-x(v))M}$ can be linearized.
For convenience, let us write for all $v \in D_{\hat{u}}^k$, $X_{\hat{u},v} = d(\hat{u}, v) + (1-x(v))M$. 

Let us introduce $n^k_{\hat{u}} = |D_{\hat{u}}^k|$ new variables $(y_{1}, .., y_{n^k_{\hat{u}}}) \in \{0,1\}^{n^k_{\hat{u}}}$ and $M'$ such that $M'$ is a upper bound of the $d(\hat{u},v)$ (we can take for instance $M' = M$). The constraint $(2)$ related to $\hat{u}$ can now be written as: $d_{acc}^{x}(\hat{u}) = \min\limits_{v \in D_{\hat{u}}^k}{X_{\hat{u},v}}$.

Let us define the following linear constraints:

$$
\begin{array}{llllll}
\label{linearization}
    & d_{acc}^{x}(\hat{u}) & \leq & X_{\hat{u},v} & \forall v \in D_{\hat{u}}\\
    & X_{\hat{u},v} - M'(1-y_v) - \sum\limits_{v' != v}{M'y_{v'}} & \leq & d_{acc}^{x}(\hat{u}) & \forall v \in D_{\hat{u}}\\
    & \sum\limits_{v \in D_{\hat{u}}^k}y_v & = & 1
    
\end{array}
$$

The last constraint induces that only one variable $y_{v^*}$ is equal to $1$, and the other ones are equal to $0$. The first set of constraint induces that $d_{acc}^{x}(\hat{u}) \leq \min\limits_{v \in D_{\hat{u}}^k}{X_{\hat{u},v}}$. The second set of constraints induces the only $v$ that verifies  $X_{\hat{u},v} \leq d_{acc}^{x}(\hat{u})$ is $v^*$. Thus, we have $d_{acc}^{x}(\hat{u}) = \min\limits_{v \in D_{\hat{u}}^k}{X_{\hat{u},v}} = X_{\hat{u},v^*}$. So we have linearized the constraint $(2)$ for $\hat{u} \in U$, thanks to the linear constraints above.

The same method can be applied for all $u\in U$, which proves the proposition.
\end{proof}

\begin{thm}

The two following propositions hold:

\begin{enumerate}[label=$\roman*)$]
    \item There is a one-to-one correspondence between the feasible solutions of $(P)$ and the feasible solutions of the OSDNP;
    \item There is a one-to-one correspondence between the optimal solutions of $(P)$ and the optimal solutions of the OSDNP.

\end{enumerate}
\end{thm}

\begin{proof}
\fbox{$i)$} Let $V_{sol} \in \mathcal{V_F}$ be a feasible solution of the OSDNP. We define $x_{V_{sol}} \in {\Z^+}^{n_t}$ a $0/1$ vector describing whether or not a bus stop is in the solution. More formally, for all $v \in V$:

$$
x_{V_{sol}}(v) = \left\{
    \begin{array}{ll}
        1 & \mbox{if } v \in V_{sol} \\
        0 & \mbox{otherwise}
    \end{array}
\right.
$$

Note that the decision variables $d_{acc}^{x}$ are entirely set by the definition of $x$ with constraint $(2)$ of $(P)$. Let us prove that $(x_{V_{sol}}, d_{acc}^{x_{V_{sol}}})$ is a feasible solution of $(P)$.

As $V_{sol} \in \mathcal{V_F}$, we know by hypothesis that 

\begin{enumerate}[label=\Alph*)]
     \item $\forall u \in U$, $d_{acc}^{V_{sol}}(u) \leq k[u]\times d_{acc}(u)$
     \item $\forall (u_1, u_2) \in U^2$, $Opt_{V_{sol}}(u_1, u_2) \leq \alpha \times Opt_{V}(u_1, u_2)$ 
    \item $|V_{sol}| \leq (1-p_{elim})|V|$
 \end{enumerate}

From $A)$, we know that for all $u \in U$, $\min_{v \in V_{sol}}{d(u,v)} \leq k[u]d_{acc}(u)$. Let $v^* \in V_{sol}$ such that $v^* = argmin_{v \in V_{sol}}{d(u,v)}$. We have $d(u, v^*) \leq k[u]d_{acc}(u)$. So we have $v^* \in V_{sol} \cap D_u^k$, so $V_{sol} \cap D_u^k$ is not an empty set and constraint $(1)$ is satisfied.

From $B)$ and equation \ref{optVsolAssumptionSimpl}, and by noticing that $d_{acc}^{x_{V_{sol}}} = d_{acc}^{V_{sol}}$ we have $d_{acc}^{V_{sol}}(u_1) + d_{acc}^{V_{sol}}(u_2) + PC[acc(u_1)][acc(u_2)] \leq \alpha(d_{acc}^{V}(u_1) + d_{acc}^{V}(u_2) + PC[acc(u_1)][acc(u_2)])$. Thus constraint $(3)$ is satisfied.

From $C)$ we know that the number of $0$ in $x_{V_{sol}}$ must be more than $p_{elim}n_t$, so constraint $(4)$ is satisfied. Thus, $(x_{V_{sol}}, d_{acc}^{x_{V_{sol}}})$ satisfies all the constraints of $(P)$, so it is a feasible solution for $(P)$. 

Reciprocally, let $x \in \mathcal{P_F}$ be a feasible solution of $(P)$. We define $V_{sol} \subseteq V$ as $V_{sol} = \{v \in V | x(v) = 1\}$. From $(1)$, we know that for all $u \in U$, $\{v \in V_{sol} | d(u,v) \leq k[u]d_{acc}(u)\} \neq \emptyset$, and {\em a fortiori} $d_{acc}^{V_{sol}}(u) \leq k[u]\times d_{acc}(u)$, and $V_{sol}$ satisfies $A)$.

As $d_{acc}^{x} = d_{acc}^{V_{sol}}$ by the definition of $V_{sol}$, the constraints $(2)$, $(3)$ and equation \ref{optVsolAssumptionSimpl} induce that $V_{sol}$ satisfies $B)$.

Just like before, constraint $(4)$ and the definition of $V_{sol}$ induce that $|V_{sol}| \leq (1-p_{elim})n_t$, $C)$ is satisfied and $(i)$ is proved.

\fbox{$ii)$} As there is a one-to-one correspondence between the feasible solutions of $(P)$ and the feasible solutions of the OSDNP, we just have to prove that the objective functions of both problems are the same. Let $x$ be a feasible solution of $(P)$ and the corresponding $V_{sol} \in \mathcal{V_F}$ defined just like before. We have 

$$
\begin{array}{ccc}
    \min\limits_{x \in \mathcal{P_F}}{\sum\limits_{(u_1, u_2) \in U^2}{OD[u_1][u_2](d^x_{acc}(u_1)+d^x_{acc}(u_2)}})  & = & \min\limits_{x \in \mathcal{P_F}}{\sum\limits_{(u_1, u_2) \in U^2}{OD[u_1][u_2](d^x_{acc}(u_1)+d^x_{acc}(u_2)} + PC[acc(u_1)][acc(u_2)]})\\
     &=& \min\limits_{V \in \mathcal{V_F}}{\sum\limits_{(u_1, u_2) \in U^2}}{OD[u_1][u_2](d^V_{acc}(u_1)+d^V_{acc}(u_2) + PC[acc(u_1)][acc(u_2)]})\\
     &=& \min\limits_{V \in \mathcal{V_F}}{\sum\limits_{(u_1, u_2)\in U^2}{OD[u_1][u_2]Opt_{V(u_1, u_2)}}}\\
     &=& \min\limits_{V \in \mathcal{V_F}}{TWT_V}
\end{array}
$$

Thus, if $V \in \mathcal{V_F}$ minimizes $TWT_V$, the corresponding $x \in \mathcal{P_F}$ minimizes $\sum\limits_{(u_1, u_2) \in U^2}{OD[u_1][u_2](d^x_{acc}(u_1)+d^x_{acc}(u_2)}$, and $ii)$ is proved, which proves the theorem.

\end{proof}

We are now able to find an optimal solution of the OSDNP by finding an optimal solution of $(P)$. Such a solution can be found using standard MILP solving algorithms. We will now apply the previous results with real data on the urban zone of Lyon, France.

\section{Case Study}
\label{caseStudy}

We give a case study of the model we described in the previous section. In this application, we consider the Lyon's urban zone, in France. We consider the bus network, which consists of $1144$ bus stops and $1464$ urban areas. A map of the corresponding urban zone is represented figure \ref{fig:all_open}

\begin{figure}[ht]
    \centering
    \includegraphics[width=8cm]{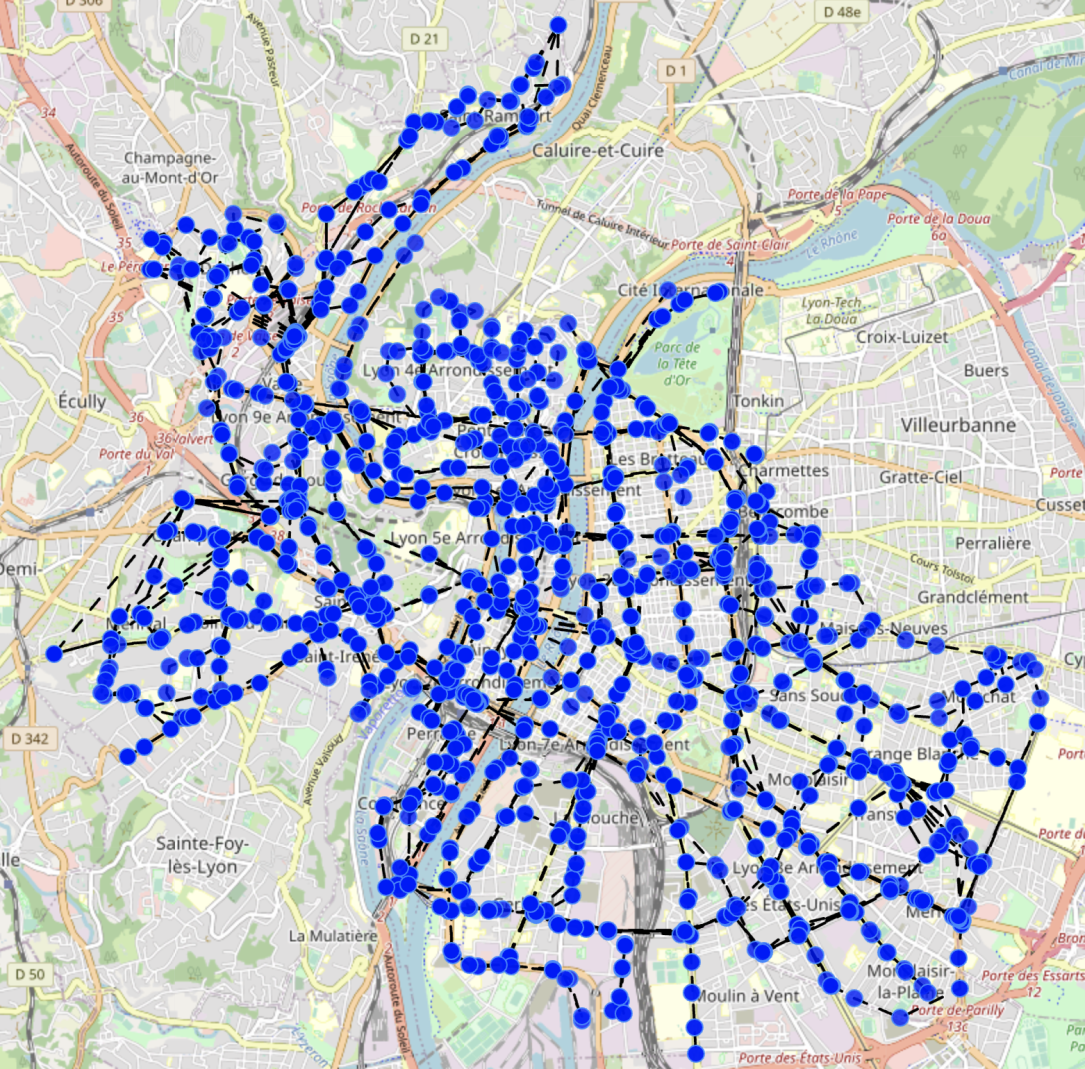}
    \caption{Map of Lyon's urban area and the $1144$ bus stops in it}
    \label{fig:all_open}
\end{figure}

We choose the values of the parameters as follow:
\begin{itemize}
    \item $\alpha = 2$
    \item for all $u \in U$, $k[u] = 2$
    \item the origin/destination matrix has been set with real observations of itineraries between the different urban areas
    \item in our application, $p_{elim}$ will take values between $0.1$ and $0.6$ (for values above $0.6$, no feasible solution exist in our case, due to the value of the other parameters)
\end{itemize}

This choice of parameter seems reasonable for real case application, since the increase of access time and travel time is acceptable for most people with this choice of parameters.

\subsection{Results}
\label{sec:Results}

We use the optimization software CPLEX \cite{cplex2009v12} to solve the MILP on the real instances described above. We represent the solution in figure \ref{fig:array_results} by representing the deleted bus stops in red while the remaining ones stay in blue. 

\begin{figure}[ht]
\begin{center}
\begin{tabular}{ccc}
  \includegraphics[width=5.5cm]{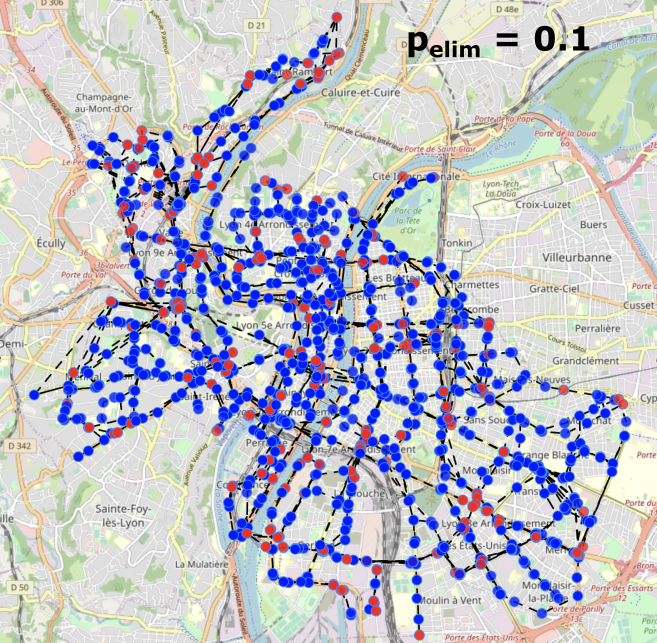}	
  &
  \includegraphics[width=5.5cm]{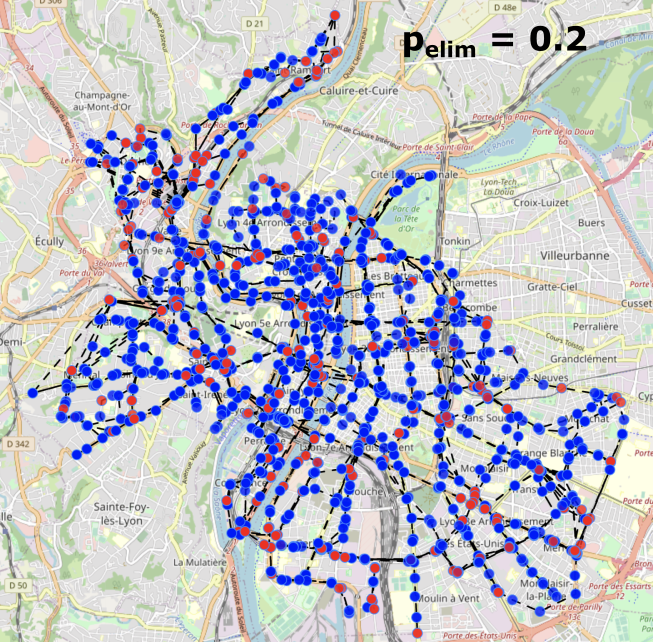}
  &
  \includegraphics[width=5.5cm]{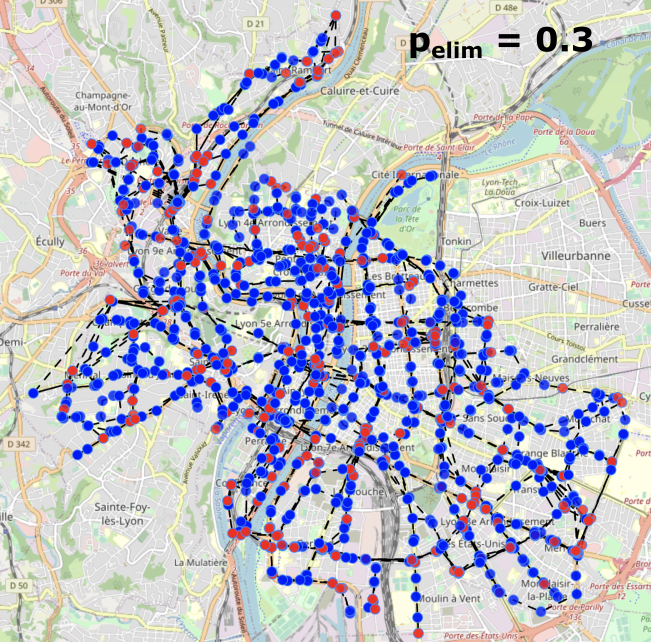}\\
  \includegraphics[width=5.5cm]{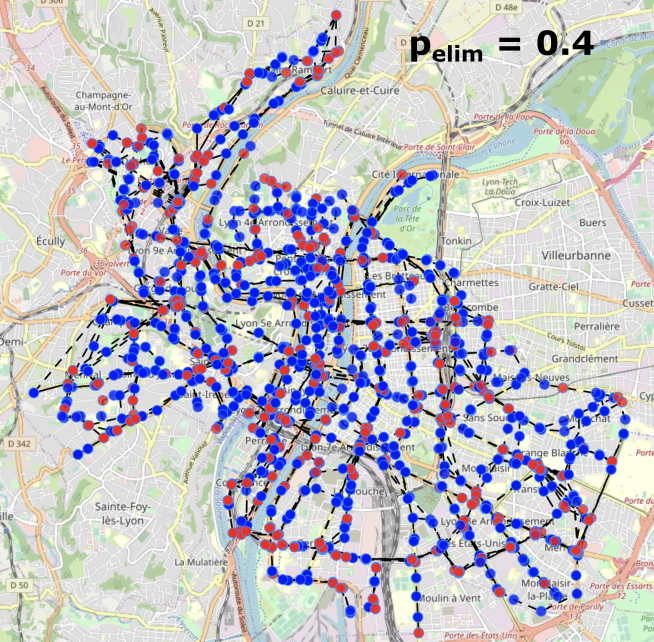}	
  &
  \includegraphics[width=5.5cm]{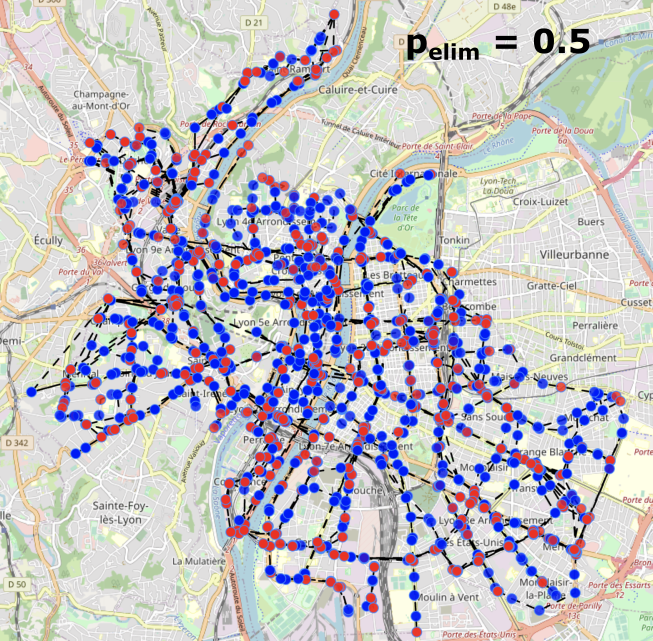}
  &
  \includegraphics[width=5.5cm]{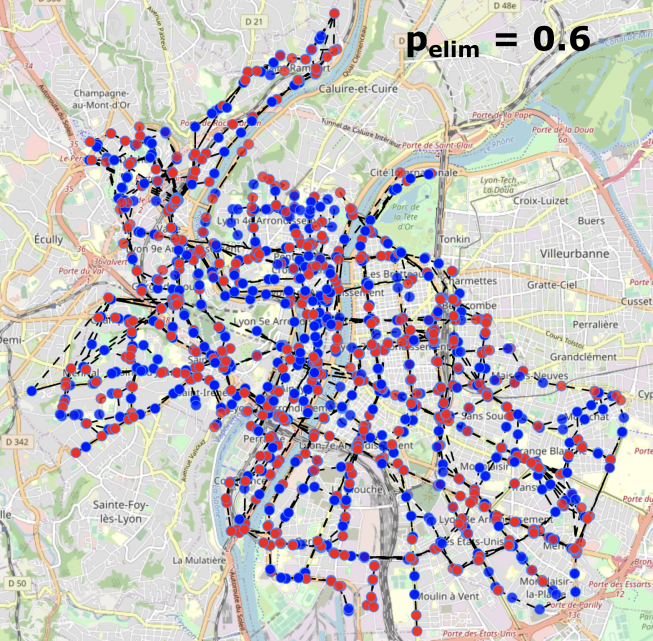}
\end{tabular}
\end{center}

  \caption{Map of remaining open stops in blue, and deleted stops in red for different values of $p_{elim}$.}
  \label{fig:array_results}
\end{figure}

One interesting thing is the number of deleted bus stops with regards to $p_{elim}$. We expect the number of deleted bus stops to be exactly equal to $\lceil p_{elim}n_t \rceil$. Figure \ref{fig:arrayNbDeleted} shows the number of deleted node with respect to $p_{elim}$

\begin{figure}[ht]
    \centering
    \label{fig:arrayNbDeleted}
\begin{tabular}{|c|c|c|}
    \hline
    $p_{elim}$ &  $\lceil p_{elim}n_t \rceil $ & \# deleted bus stops \\
    \hline
     $0.1$ & $115$ & $266$ \\
     \hline
     $0.2$ & $229$ & $295$ \\
     \hline
     $0.3$ & $344$ & $344$ \\
     \hline
     $0.4$ & $458$ & $458$ \\
     \hline
     $0.5$ & $572$ & $572$ \\
     \hline
     $0.6$ & $687$ & $687$ \\
     \hline
\end{tabular}
\label{arr:nbDeleted}
\end{figure}

We notice that for $p_{elim} = 0.1$ and $p_{elim} = 0.2$, we delete more stops that we are imposed to. This comes from the fact that for low values of $p_{elim}$ the optimal solution does not need that much number of bus stops, because the number of urban areas is limited with regards to the number of bus stops. Increasing the number of urban zones ($n_u$) would solve this side-effect. However, this would also increase the computational time needed to find a solution.

\subsection{Analysis}
\label{sec:Analysis}

Now that we have solved instances for several values of $p_{elim}$, we would like to evaluate the impact of such solutions on the network with an operational point of view. For a public transportation network designer, the number of lines is a very important input. The number of buses, the number of bus drivers, the complexity and the number of the transit times are directly linked to it. Consequently, even if we insure a decrease of the number of stops, this does not induce such a decrease (or a decrease at all) of the number of lines. Let us analyze the number of stops that are still open for each line in order to be able to build {\em an operational decision tool}. Ideally, we want this percentages to be either close to $0$ or close to $1$. If it is the case, we can keep the lines which percentage is close to $1$ and delete the other lines. 

For some lines, the percentage of remaining open stop can drastically change. We give an example of the possible differences for $p_{elim} = 0.5$ in figure \ref{fig:pourcOpen0.5}

\begin{figure}[ht]
\begin{center}
\begin{tabular}{ccc}
  \includegraphics[height=5.5cm]{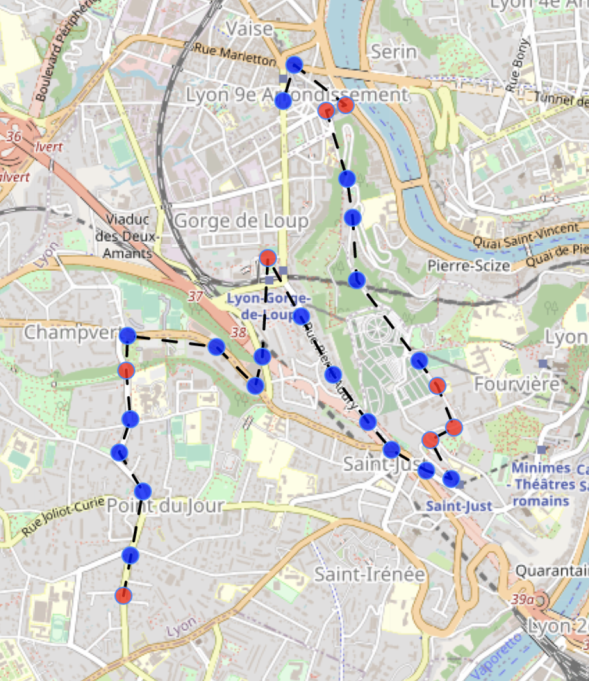}	
  &
  \includegraphics[height=5.5cm]{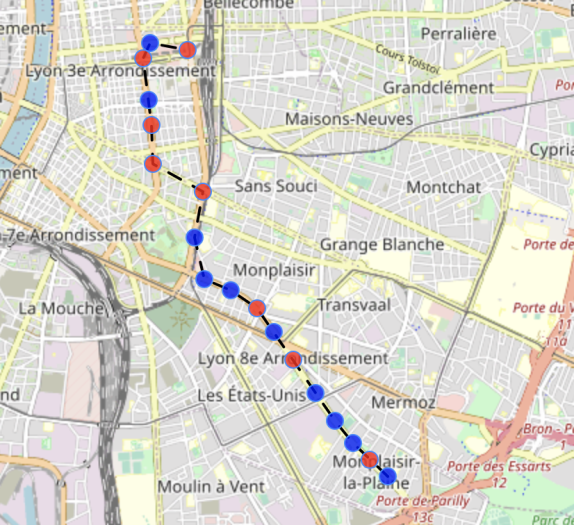}
  &
  \includegraphics[height=5.5cm]{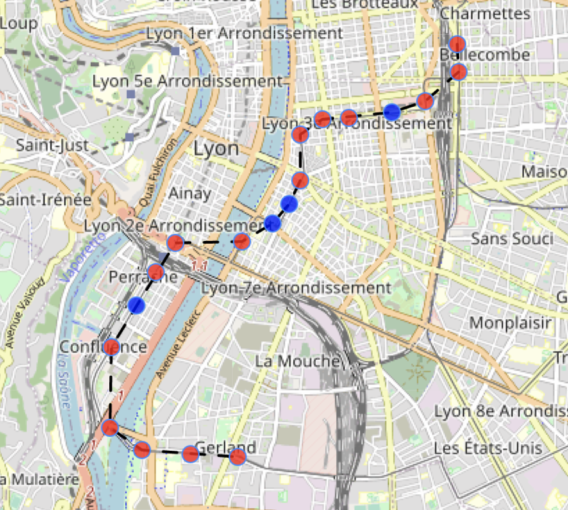}
\end{tabular}
\end{center}
\label{fig:pourcOpen0.5}
\caption{Lines 90, C25 and T1 (or some sublines) with respectively $71\%$, $55.5\%$ and $21\%$ of remaining open stops}
\end{figure}

So the profiles of the lines can be different. We can plot the histogram of the percentages of remaining open stops for every lines that contain more than 10 bus stops, (we keep only such lines not to have too much side-effects). We represent such an histogram for $p_{elim} = 0.5$ in figure \ref{fig:hist0.5}.

\begin{figure}[ht]
    \centering
    \includegraphics[width = 13cm]{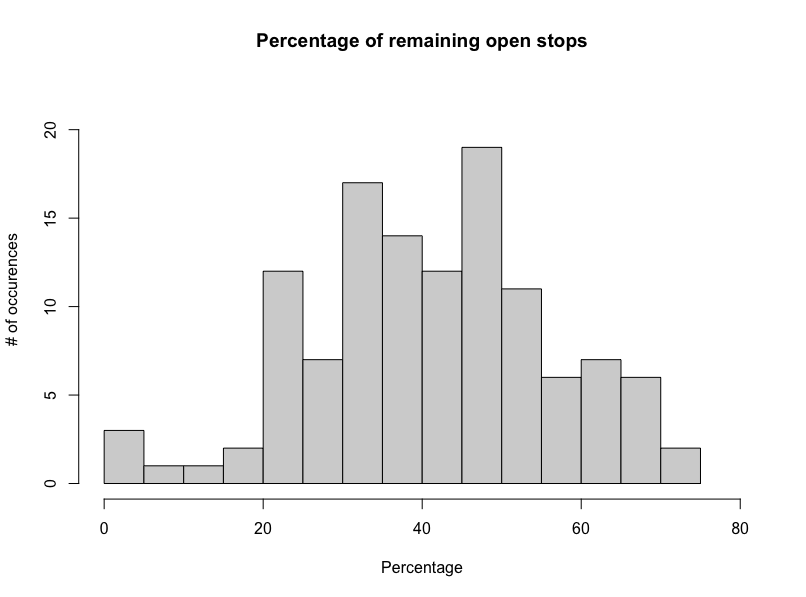}
    \caption{Histogram of the percentage of remaining open stops for every line containing more than 10 stops for $p_{elim} = 0.5$}
    \label{fig:hist0.5}
\end{figure}

First of all, we can see that the ideal case is not reached, since most values are close to the average value. Even if we cannot conclude that the percentage of open stops follows a normal law since the p-value of the Shapiro-Wilk test is $p = 0.2072$, most values are neither close to $0$ nor close to $1$.

However, we still want to help the network designers to know which lines to keep and which ones to delete. To do so, we give him some scenarios based on a {\em threshold $t \in [0,1]$}. This threshold represent the percentage above which we will keep the lines.

More formally, let $l = \{s_1, s_2, .., s_{|l|}\}$ a line which is a sequence of stops ($s_1, .., s_{|l|} \in V$). For an optimal solution $V^*_{sol}$, we define the percentage of remaining open stops of $l$ as $p_{ros}^{V^*_{sol}}(l) = \frac{|V^*_{sol}|}{|l|}$. Let $L$ be the set of all lines. We define $L_{k}^t$ ($k$ for keep) and $L_{d}^t$ ($d$ for deleted) a partition of $L$ such that $L_{k}^t = \{l\in L | p_{ros}^{V^*_{sol}}(l) \geq t\}$ and $L_{d}^t = \{l\in L | p_{ros}^{V^*_{sol}}(l) < t\}$. The strategy here is, for several values of $t$, to evaluate the sets $L_{k}^t$ and $L_{d}^t$. If $t$ is close to $0$, then $L_{k}^t$ will be close to $L$, and the percentage of lines deleted will be small, maybe too small for the network designer. If $t$ is close to $1$, then $L_{d}^t$ will be close to $L$ and the network efficiency will be widely degraded. So we propose different scenarios to the network designer, who will choose one of them with regards to the trade-off between cost and efficiency her prefers.

To build a scenario $\mathcal{S}$, we begin with choose a $p_{elim}$ and we compute an optimal solution $V^*_{sol}$ thanks to our model described in section \ref{MILP}. Then we choose $t \in [0,1]$ and compute the corresponding $L_{k}^t$ and $L_{d}^t$. We arise with a new solution (which is not feasible in general) $V_{\mathcal{S}} = V_{sol}^* \setminus L_{d}^t$, which contains the stops in $V_{sol}^*$, minus the stops of the lines in $L_{d}^t$. On the one hand, we have deleted exactly $|L_{d}^t|$ lines from $L$, which would help the network designer to manage the network. On the other hand, since the solution is no longer feasible (in general), there will be some $u \in U$ for which the access time will be too long with regards to our original constraint (constraint $(2)$ of $(P)$). Let $UF_\mathcal{S}$ be the set of such $u \in U$. Let us call $uf(u) = k[u]\times d_{acc}(u) - d_{acc}^{V_{\mathcal{S}}}(u)$, where $d_{acc}^{V_{\mathcal{S}}}(u)$ is the minimal access time to a bus stop in $V_{\mathcal{S}}$: $d_{acc}^{V_{\mathcal{S}}}(u) = \min_{v \in V_{\mathcal{S}}}{d(u, v)}$. Then $UF_\mathcal{S} = \{u \in U | uf(u) < 0\}$. From a scenario $\mathcal{S}$, we also give the histogram $\mathcal{H}_{uf}$ of the $uf(u)$ which can be an interesting input for the network designer.  We describe schematically the construction of a scenario $\mathcal{S}$ in figure \ref{fig:scenarioConstuction}

\begin{figure}[ht]
    \centering
    \includegraphics[width = 15cm]{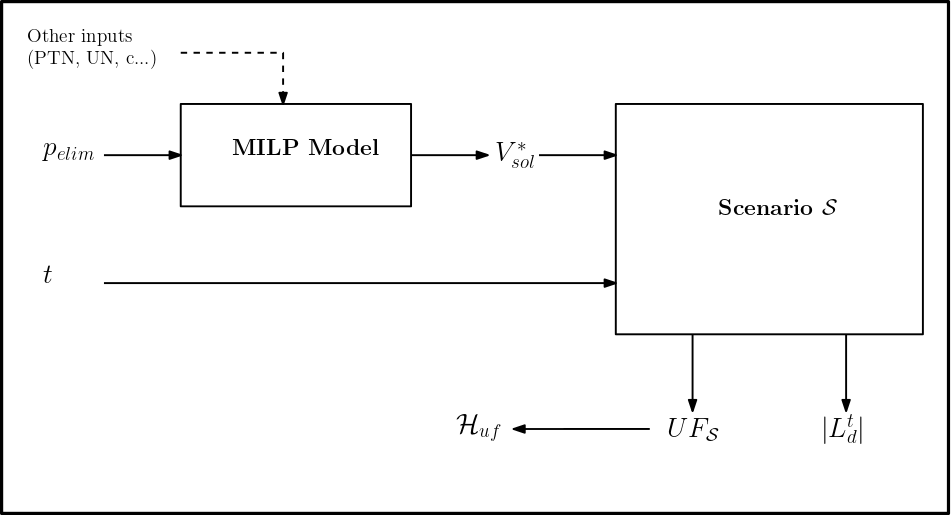}
    \caption{Schema of the construction of a scenario $\mathcal{S}$}
    \label{fig:scenarioConstuction}
\end{figure}

We present $7$ scenarios for $p_{elim} = 0.5$ and $t$ values from $0.1$ to $0.7$ in figures \ref{fig:resultsArray1} and \ref{fig:resultsArray2}. In this instance, we have $120$ different lines containing more than $10$ stops, and the number of urban zones is $n_u = 1464$.

\begin{figure}[ht]
\begin{center}
\begin{tabular}{|c|c|c|c|}
   \hline
   $t$ & $|L_d^t|$ & $|UF_\mathcal{S}|$ & $\mathcal{H}_{uf}$ \\
   \hline
   $0.1$ & $0$ & $0$ & $\emptyset$ \\ 
   \hline
   $0.2$ & $3$ & $2$ & \includegraphics[width = 9.6cm]{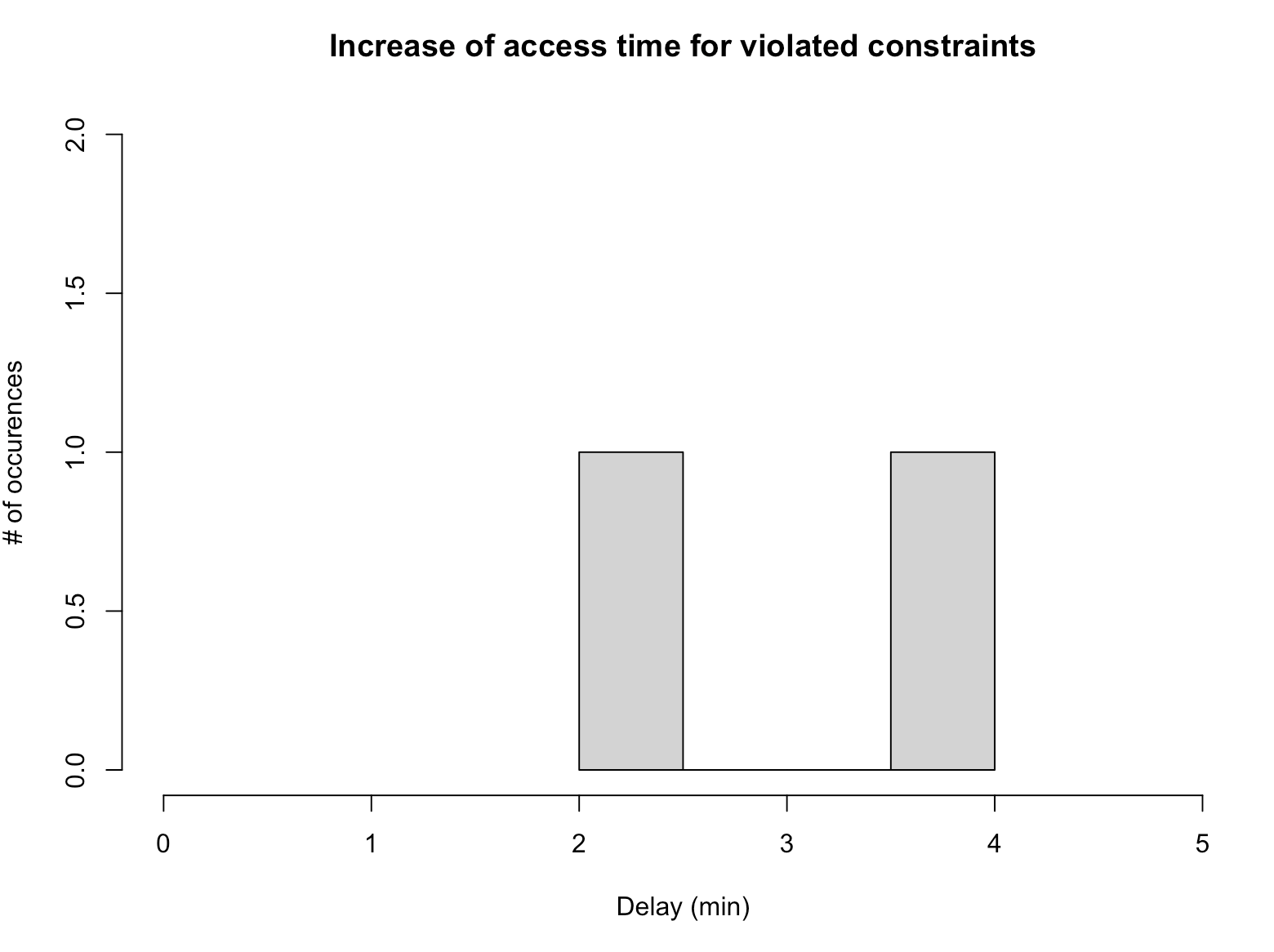} \\ 
   \hline
   $0.3$ & $9$ & $10$ & \includegraphics[width = 9.6cm]{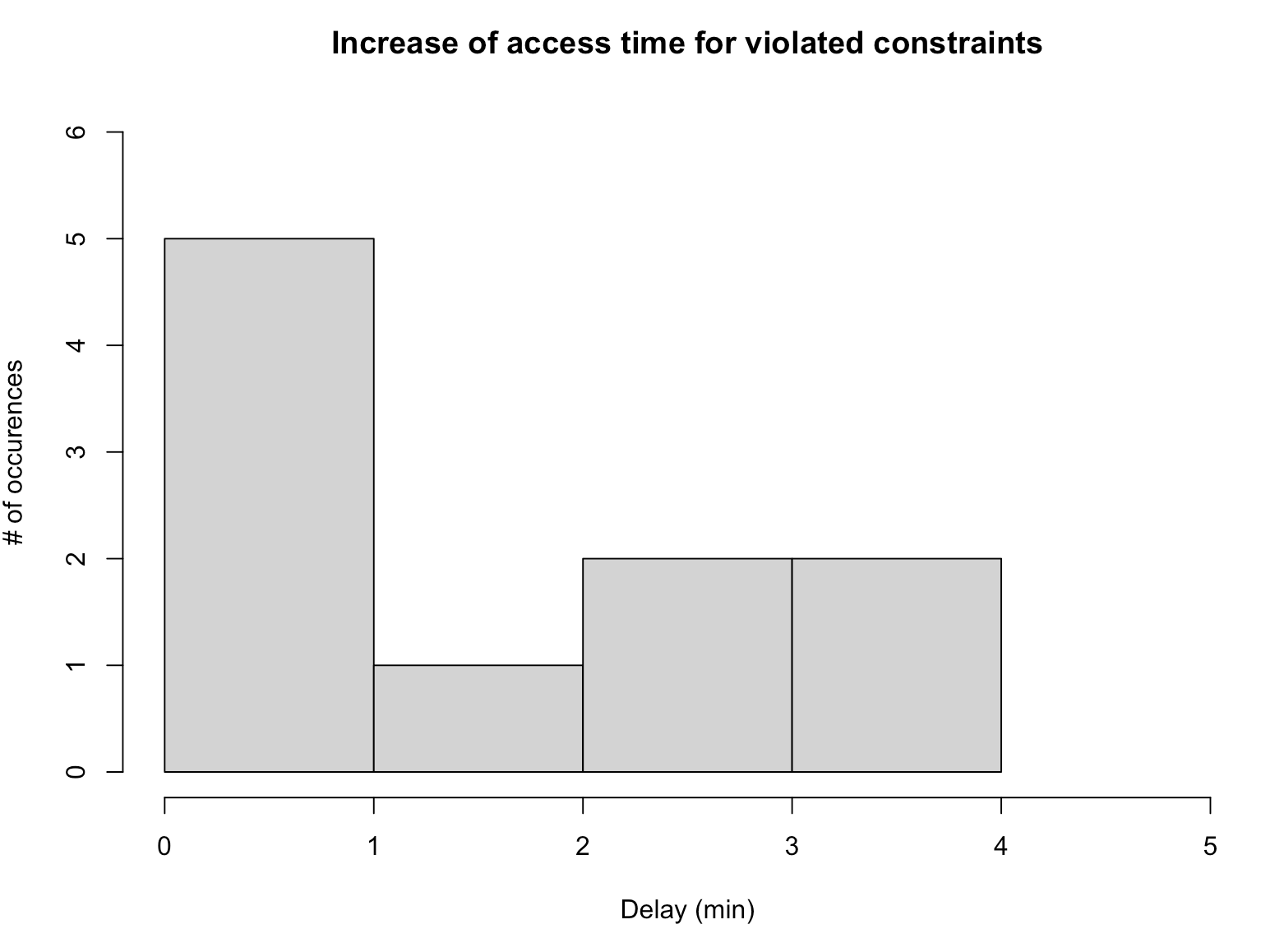} \\ 
   \hline
   $0.4$ & $35$ & $123$ & \includegraphics[width = 9.6cm]{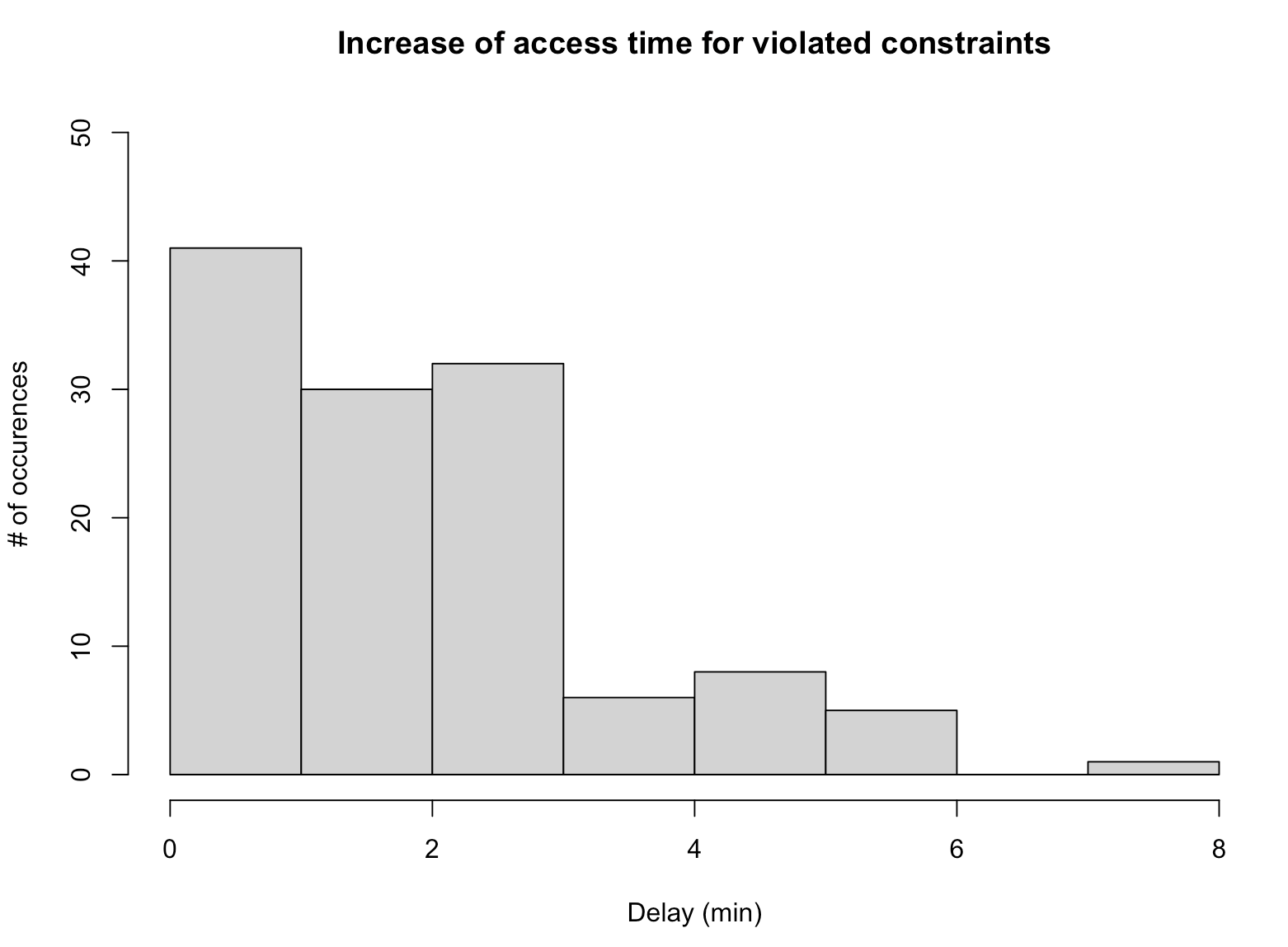} \\
      \hline
\end{tabular}
\end{center}
\caption{$|L_d^t|$, $|UF_\mathcal{S}|$ and $\mathcal{H}_{uf}$ for values of $t$ between $0.1$ and $0.4$ and $p_{elim} = 0.5$}
\label{fig:resultsArray1}
\end{figure}

\begin{figure}[ht]
\begin{center}
\begin{tabular}{|c|c|c|c|}
\hline
   $t$ & $|L_d^t|$ & $|UF_\mathcal{S}|$ & $\mathcal{H}_{uf}$ \\
   \hline
   $0.5$ & $93$ & $572$ & \includegraphics[width = 9.6cm]{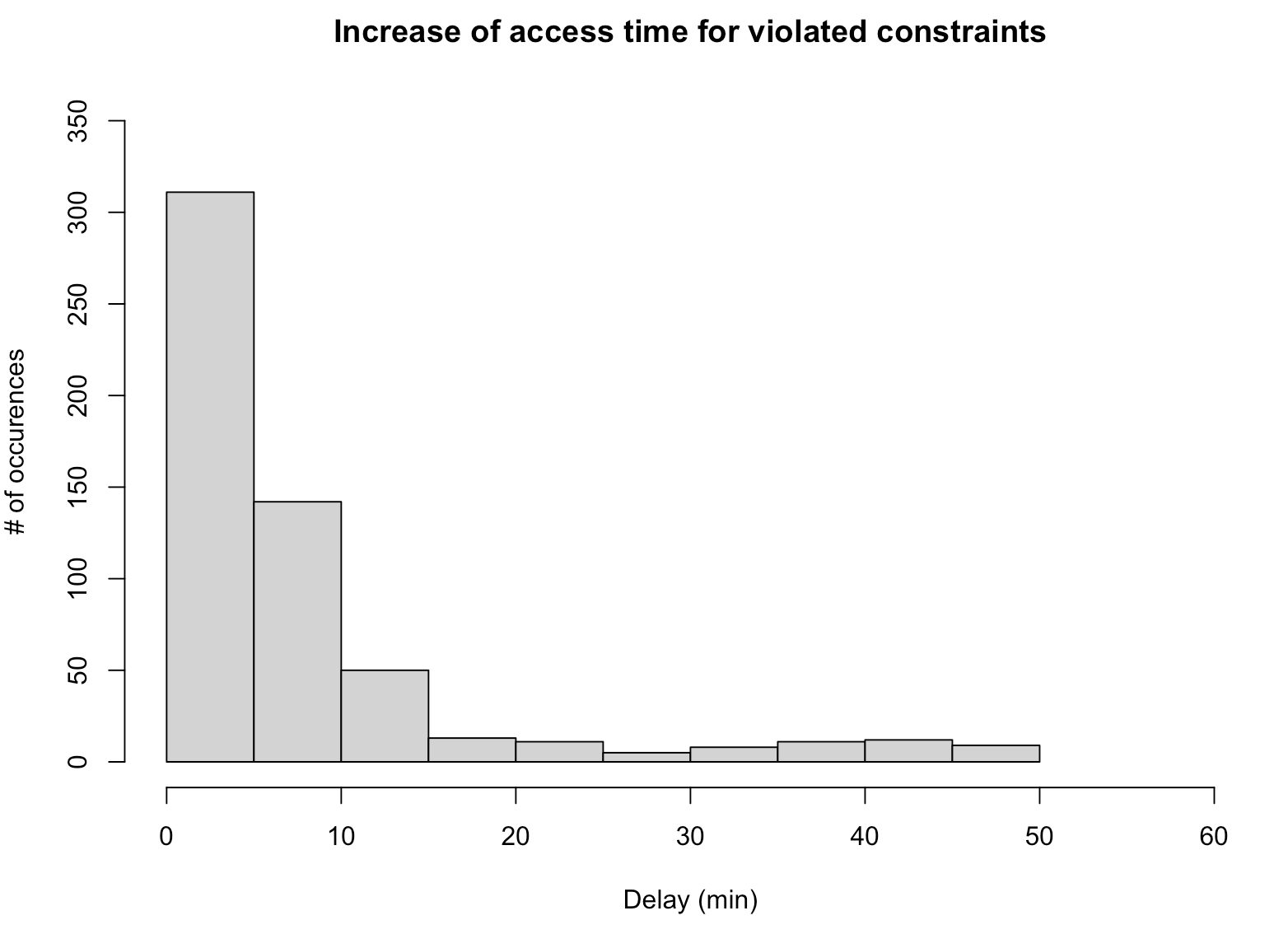} \\ 
   \hline
   $0.6$ & $118$ & $867$ & \includegraphics[width = 9.6cm]{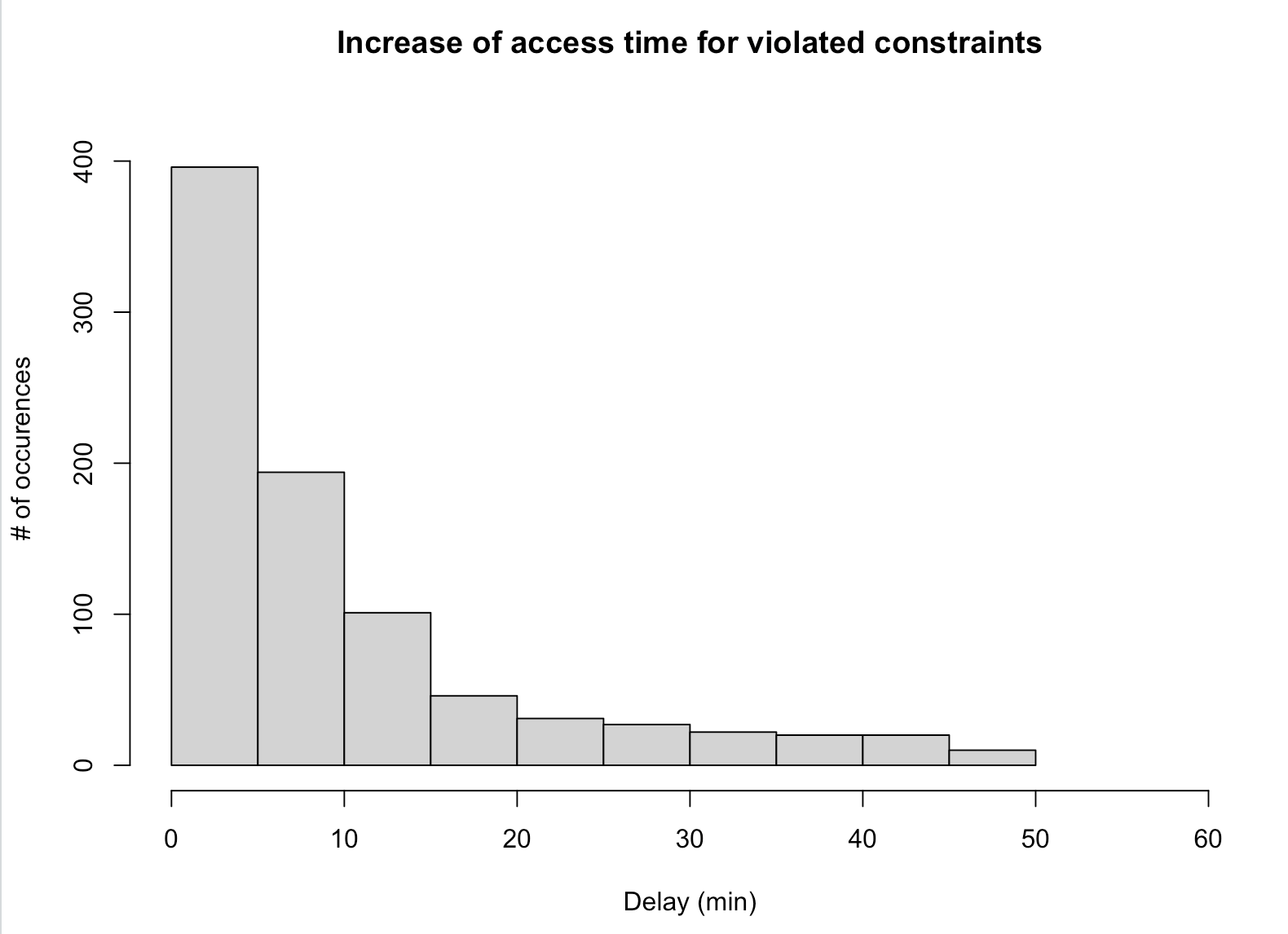} \\ 
   \hline
   $0.7$ & $120$ & $938$ & \includegraphics[width = 9.6cm]{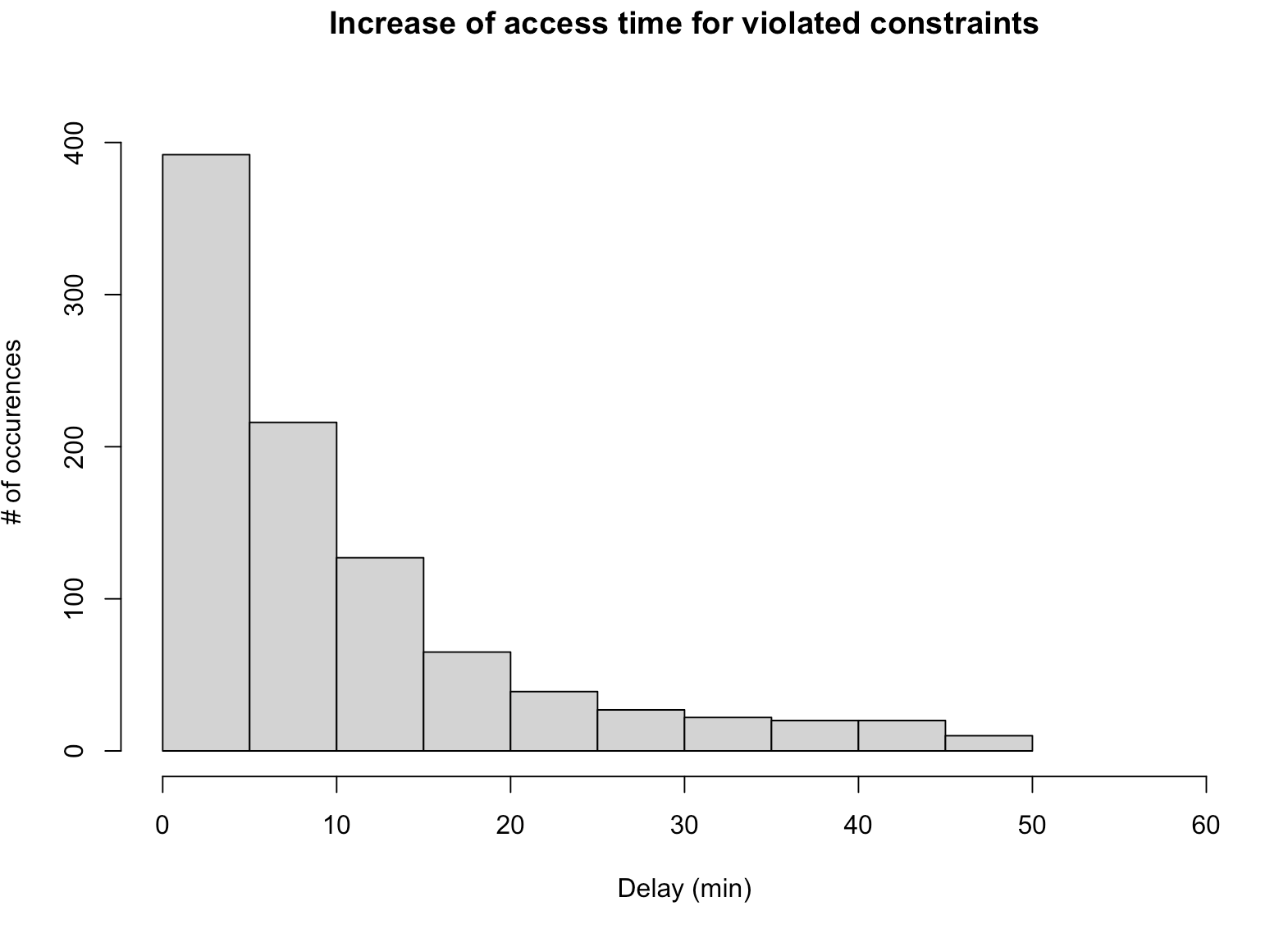}  \\
   \hline
\end{tabular}
\end{center}
\caption{$|L_d^t|$, $|UF_\mathcal{S}|$ and $\mathcal{H}_{uf}$ for values of $t$ between $0.5$ and $0.7$ and $p_{elim} = 0.5$}
\label{fig:resultsArray2}
\end{figure}

Note that even for $t$ close to $1$, not all urban zones are violated because there are still some stops open on bus lines that contain less than $10$ stops. Moreover, if we increase the value of $t$ ($t = 0.8$, $t = 0.9$..) we obtain the same results than for $t = 0.7$.

In these scenarios, we see that when $t$ approaches the value of $p_{elim}$, both $|L_d^t|$ and $|UF_\mathcal{S}|$ strongly increase, which is coherent with the fact that the percentage of remaining open stops of the lines are mostly around $p_{elim}$.

With these scenarios, a network designer can evaluate the different cases and choose the one which is likely to match his wishes. 

\section{Conclusion and Future Work}
\label{conclusion}

In this paper we have been interested in the highlighting of a subgraph during disturbed conditions. After defining the problem, we have given a Mixed Integer Linear Program modeling the problem. We have solved the model on a real-case application on the Lyon's urban zone. We have presented the results for a reasonable choice of parameters, and for different values of deletion rate. We also have given a decision tool that could be useful for the network designers to choose his best trade-off between costs and efficiency. 

The computational running time is a big issue in our article. Indeed, we made some simplifying hypothesis in order to allow the real-case resolution. A more realistic MILP could have been written (and has been), but its resolution were in our case too long for our real-case application. To go further, it could be interesting to dig into more complex resolution algorithms to get rid of simplifying hypothesis and still get results on real instances. 

Note that the choice of parameters is a real issue with regards to the computational running time too. Indeed, if we take for instance our parameter $k$ (which describes how far the accession time is allowed to increase), then if $k$ increases, the number of potential bus stops to look into can increase, making the corresponding instances hard to solve in reasonable time.

The guideline of this paper has been to highlight a choice of bus stops in order to choose a subset of lines at the end. However, the choice of such lines induces in the general case a set of unfeasible bus stops according to our MILP formulation. Even if we have been able to evaluate the profile of the unfeasibility, we have no guarantee {\em a priori} of the quality of our solution in term of bus lines. This could be a problem in some cases. Another solution could have been to see the problem with a "line" point of view from the beginning, and to define a solution as a subset of lines to  keep. However, this could increase drastically the computational complexity.

\clearpage

\printbibliography

\end{document}